\documentclass[reprint,superscriptaddress,amsmath,amssymb,aps,pra]{revtex4-2}

\usepackage{graphicx}
\usepackage{dcolumn}
\usepackage{bm}
\usepackage{physics}
\usepackage{amssymb}
\usepackage{textcomp}
\usepackage{amsthm}
\usepackage{mathtools}
\usepackage[utf8]{inputenc}
\usepackage[czech,polish,english]{babel}
\usepackage{MnSymbol}
\usepackage{dsfont}
\usepackage[shortlabels]{enumitem}
\usepackage[T1]{fontenc}
\usepackage{xcolor}
\usepackage{hyperref}
\usepackage{orcidlink}
\hypersetup{
    colorlinks,
    citecolor=blue,
    filecolor=black,
    linkcolor=blue,
    urlcolor=blue
}

\DeclarePairedDelimiter{\Ket}{|}{\rrangle}

\newcommand\BraKet[2]{{\llangle}{#1}{|}{#2}{\rrangle}}

\newcounter{definition}
\newenvironment{definition}[1][]{\refstepcounter{definition}\par
   \textit{Definition~\thedefinition.} }

\newcounter{theorem}
\newenvironment{theorem}[1][]{\refstepcounter{theorem}
   \textit{Theorem~\thetheorem.} }

\newcounter{proposition}
\newenvironment{proposition}[1][]{\refstepcounter{proposition}
   \textit{Proposition~\theproposition.} }

\newcounter{example}
\newenvironment{example}[1][]{\refstepcounter{example}
   \textit{Example~\theexample.} }

\renewenvironment{proof}{
   \textit{Proof.} }{\qed}

\begin{document}

\title{Stories in the two-state vector formalism}

\author{Patryk Michalski\orcidlink{0009-0009-0305-7356}}
\affiliation{Institute of Theoretical Physics, University of Warsaw, Pasteura 5, 02-093 Warsaw, Poland}

\author{Andrzej Dragan\orcidlink{0000-0002-5254-710X}}
\email{dragan@fuw.edu.pl}
\affiliation{Institute of Theoretical Physics, University of Warsaw, Pasteura 5, 02-093 Warsaw, Poland}
\affiliation{Centre for Quantum Technologies, National University of Singapore, 3 Science Drive 2, 117543 Singapore, Singapore}

\date{\today}

\begin{abstract}
The two-state vector formalism of quantum mechanics is a time-symmetrized approach to standard quantum theory. In our work, we aim to establish rigorous foundations for the future investigation within this formalism. We introduce the concept of a story --- a compatible pair consisting of a two-state vector and an ideal measurement. Using this concept, we examine the structure of the space comprising all two-state vectors. We analyze the problem of distinguishability and confirm that some pairs of two-state vectors or their statistical mixtures cannot be physically distinguished. In particular, we discuss an example of a two-state vector that is indistinguishable from a statistical mixture of separable two-state vectors and provide an example of a two-state vector that can be distinguished from every such mixture. This leads us to formulate the definition of a strictly non-separable two-state vector as a genuine manifestation of entanglement between the past and the future. 
\end{abstract}

\maketitle

\section{Introduction}
One of the distinguishing features of conventional quantum mechanics is the time asymmetry of ideal measurements, which contrasts with the time-symmetric nature of unitary evolution \cite{gmh_1994}. Every pure quantum state determines only the probability distribution of future measurement results. Thus, based on a single measurement outcome, it is impossible to recover the state in which the system was before the measurement took place. At the same time, the results of past measurements clearly define the present state of the system. This differs from classical physics, where there is no distinction between past and future measurements. 

In 1964, Aharonov, Bergmann and Lebowitz proved that the apparent time asymmetry is not an inherent property of quantum theory, as it can be eliminated by introducing a modified description of ideal measurements \cite{abl_1964}. Building on this observation, a novel approach to quantum mechanics was developed \cite{av_1991}. It was called the \emph{two-state vector formalism} in accordance with its central concept --- the two-state vector, which describes the state of a system using pairs of forward- and backward-evolving ordinary state vectors. Since its introduction, this formalism has led to the exploration of numerous peculiar effects \cite{av_1991,vaa_1987,v_1998,av_2003,v_1993,mrv_2007,aav_1988,av_1990,dfv_2013}, evoking some controversies as to their interpretation \cite{l_1989,p_1989,fc_2014,s_2016,sa_2018,s_2018,h_2016,laz_2013}. Recently, the necessity of adopting a framework consistent with the two-state vector formalism has been demonstrated in an attempt to construct a covariant quantum field theory of tachyons, raising a number of interesting questions to be addressed \cite{d_2023}.

During the considerable time that has passed since the development of the two-state vector formalism, several works have aimed to investigate its underlying mathematical structure \cite{ra_1995, aptv_2009, s_2014}. Nevertheless, at least one technical aspect of the original formulation has yet to be addressed satisfactorily. The purpose of this work is to draw attention to this aspect and use it as a basis to establish a rigorous framework for investigation within the formalism. It should be emphasized that our focus is primarily on highlighting parallels and discrepancies with standard quantum theory. Consequently, various generalizations, such as weak measurements or POVMs, fall outside the scope of this study.

In the following sections, new concepts necessary for further considerations are defined and crucial results are presented. We discuss the concept of a story and analyze the structure of the twin space, which is the set of all possible two-state vectors, with respect to the ability of its elements to form a story with certain measurements. Next, from the new perspective, we address the issue of the distinguishability of two-state vectors and highlight clear differences from the standard quantum formalism. Finally, we focus on the distinguishability of non-separable and separable two-state vectors in the context of hypothetical entanglement between the past and the future.

\section{Formal description of two-state vectors}

For a system represented by a Hilbert space $\mathcal{H}$, all postulated two-state vectors $\Ket{\Psi}$ at a given time $t$ can be identified with rays in a twin space $\mathcal{H} \otimes \mathcal{H}^\star$ (with $\mathcal{H}^\star$ denoting the space dual to $\mathcal{H}$):
\begin{equation}\label{eq:TSV}
    \Ket{\Psi} = \sum_k \alpha_k \ket{\psi_k} \otimes \bra{\phi_k} \in \mathcal{H} \otimes \mathcal{H}^\star.
\end{equation}
The vectors $\ket{\psi_k}$ represent forward-evolving quantum states, determined by a \emph{pre-selection} measurement performed at some time prior to $t$. Conversely, the vectors $\bra{\phi_k}$ represent backward-evolving states, arising from a \emph{post-selection} measurement performed after time $t$. These notions extend the standard quantum formalism by introducing time-symmetric boundary conditions.
Without additional constraints, the vectors $\ket{\psi_k}$ and $\bra{\phi_k}$ are arbitrary elements of the respective Hilbert spaces.

While in the simplest case a two-state vector takes the form of a single pair $\ket{\psi} \otimes \bra{\phi}$, in general one must consider linear combinations of such product states, as in Equation~\eqref{eq:TSV}. Such superpositions arise naturally when the system of interest is entangled with an ancillary subsystem \cite{av_1991}, or in relativistic settings when the state of a superluminal particle undergoes a Lorentz transformation \cite{d_2023}.

The twin space is equipped with a canonical inner product, which for the most general vectors $\Ket{\Psi}, \Ket{\Psi'} \in \mathcal{H} \otimes \mathcal{H}^\star$ has the form:
\begin{equation}
    \BraKet{\Psi}{\Psi'} = \sum_{k,\,l} \bar{\alpha}_k \alpha'_l \braket{\psi_k}{\psi'_l} \braket{\phi'_l}{\phi_k} .
\end{equation}
Moreover, there is a canonical isomorphism of $\mathcal{H} \otimes \mathcal{H}^\star$ with the algebra of Hilbert-Schmidt operators, so the elements of the twin space can be treated as linear operators with conventionally defined trace functional $\Tr: \mathcal{H} \otimes \mathcal{H}^\star \to \mathbb{C}$. For every two-state vector $\Ket{\Psi} \in \mathcal{H} \otimes \mathcal{H}^\star$, this functional yields:
\begin{equation}
    \Tr{\Ket{\Psi}} = \sum_k \alpha_k \braket{\phi_k}{\psi_k}.
\end{equation}
Although one may introduce a normalization convention for two-state vectors using the canonical inner product, in applications presented in this work the normalization factor is insignificant.

For a quantum system described by a two-state vector $\Ket{\Psi} \in \mathcal{H} \otimes \mathcal{H}^\star$, one of the primary problems to which the two-state vector formalism can be applied is finding the probabilities for distinct outcomes of an ideal projective measurement. Notice that every observable $\hat{C}$ defines a projective partition of unity $\{\hat{P}_i\}$, where each operator $\hat{P}_i$ is a projection onto a subspace corresponding to a distinguishable measurement outcome $c_i$. From now on, we will treat the set $\{\hat{P}_i\}$ as a sufficient characterization of the measurement.
The formula for the probability of an outcome $c_n$, also known as the the Aharonov-Bergmann-Lebowitz rule, takes the form \cite{av_1991}:
\begin{equation}\label{ABL}
    \mathrm{Prob}(c_n) = \frac{\left|\Tr{\hat{P}_n \Ket{\Psi}}\right|^2}{\sum_i \left|\Tr{\hat{P}_i \Ket{\Psi}}\right|^2}.
\end{equation}
From the derivation of Equation~\eqref{ABL}, it follows that when the denominator on the right-hand side is zero, the process leading to the description of the system by a two-state vector $\Ket{\Psi}$ cannot occur with a given measurement. This observation is a key point of our work.

\section{Story}

Previous research on the two-state vector formalism (or equivalently, the twin space formalism) overlooked an important fact that must be acknowledged in order to construct a consistent theory. Given an arbitrary measurement $\{\hat{P}_i\}$ within a space $\mathcal{H}$, certain pairs of pre- and post-selected system states are not obtainable. To demonstrate this, notice that every projection operator $\hat{P}_i$ can be represented as:
\begin{equation}\label{projection_decomposition}
    \hat{P}_i = \sum_{j \,\in\, J_i} \ketbra{\xi_j},
\end{equation}
where $J_i$ is a complete subset of indices that enumerate eigenstates of $\hat{C}$ with eigenvalue $c_i$, and $\{\ket{\xi_j}\}$ forms an orthonormal basis of the space $\mathcal{H}$. Suppose we attempt to construct a two-state vector $\Ket{\Psi} = \ket{\xi_j} \otimes \bra{ \xi_k}$ using two distinct elements of this basis. Then, for any $\hat{P}_i$, the trace $\Tr{\hat{P}_i \Ket{\Psi}}$ vanishes, and the denominator on the right-hand side of Equation~\eqref{ABL} is zero. This indicates that, in the considered scenario, two distinct elements of the basis cannot represent the evolved pre- and post-selected system states respectively.

Consequently, the possibility of considering a two-state vector $\Ket{\Psi} \in \mathcal{H} \otimes \mathcal{H}^\star$ is fully determined by the choice of the measurement $\{\hat{P}_i\}$, in contrast to the case of conventional state vectors. This justifies the introduction of a new notion for compatible pairs $(\Ket{\Psi},\{\hat{P}_i\})$. 
\begin{definition}
    A pair $(\Ket{\Psi},\{\hat{P}_i\})$ consisting of a two-state vector $\Ket{\Psi} \in \mathcal{H} \otimes \mathcal{H}^\star$ and a measurement $\{\hat{P}_i\}$ forms a \emph{story} if for at least one index $i$ the following condition is satisfied:
    \begin{equation}\label{eq:story}
        \left| \Tr{\hat{P}_i \Ket{\Psi}} \right| > 0.
    \end{equation}
\end{definition}
The condition~\eqref{eq:story} guarantees that an appropriate pair of evolved pre- and post-selected system states can be obtained, leading to the description of the system by the two-state vector $\Ket{\Psi}$.

\section{Structure of the twin space}\label{Structure_of_the_TS}

From the above considerations it follows that a story is the fundamental concept that holds physical significance within the discussed formalism. This observation motivates a comprehensive analysis of the twin space structure regarding the possibility of forming stories with two-state vectors. We commence by stating that, in order to build up every possible story, we need to consider the whole twin space $\mathcal{H} \otimes \mathcal{H}^\star$. In other words, there are no two-state vectors $\Ket{\Psi} \in \mathcal{H} \otimes \mathcal{H}^\star$ which do not form a story with any  measurement.

\begin{theorem}
\label{th1}
    For every vector $\Ket{\Psi} \in \mathcal{H} \otimes \mathcal{H}^\star$, there exists a measurement $\{\hat{P}_i\}$ such that the pair $(\Ket{\Psi},\{\hat{P}_i\})$ forms a story.
\end{theorem}
\begin{proof}
    Let $\{\ket{k}\}$ denote an orthonormal basis of the space $\mathcal{H}$. Every vector $\Ket{\Psi} \in \mathcal{H} \otimes \mathcal{H}^\star$ can be decomposed in the induced product basis $\{\ket{k} \otimes \bra{l}\}$ of the twin space:
    \begin{equation}
        \Ket{\Psi} = \sum_{k,\,l} \alpha_{kl} \ket{k} \otimes \bra{l}.
    \end{equation}
    We proceed with proof by exhaustion.

    Case 1: When $\exists_{n} \,\, \alpha_{nn} \ne 0$, the vector $\Ket{\Psi}$ forms a story with a measurement for which one of the elements of the projective partition of unity is the projection onto a state $\ket{n}$.
    
    Case 2: When $\forall_{k,\,l} \,\, \alpha_{kl} = - \alpha_{lk}$, the vector $\Ket{\Psi}$ forms a story with a measurement for which one of the elements of the projective partition of unity is the projection onto a state $\frac{1}{\sqrt{2}} (\ket{m} + i \ket{n})$, where $m \ne n$ and $\alpha_{mn} \ne 0$. This follows from the inequality:
        \begin{multline}
            \left|\Tr{\frac{1}{2} (\ket{m} + i \ket{n}) (\bra{m} - i \bra{n}) \Ket{\Psi}}\right| =\\
            = \frac{1}{2}\left|\sum_{k,\,l} \alpha_{kl} \bra{l} \left(\ketbra{m} - i \ketbra{m}{n} + i \ketbra{n}{m} + \ketbra{n}\right)\ket{k} \right|
             \\ =
            \frac{1}{2}\left| i (\alpha_{mn} - \alpha_{nm}) \right| = | \alpha_{mn} | > 0.
        \end{multline}
    
    Case 3: When $\exists_{m,\,n; \,\, m\ne n}$ $ \alpha_{mn} \ne - \alpha_{nm}$, $\alpha_{mn} \ne 0$ and $\forall_{k} \,\, \alpha_{kk} = 0$, the vector $\Ket{\Psi}$ forms a story with a measurement for which one of the elements of the projective partition of unity is the projection onto a state $\frac{1}{\sqrt{2}} (\ket{m} + \ket{n})$. This follows from the inequality:
        \begin{multline}
            \left|\Tr{\frac{1}{2} (\ket{m} + \ket{n}) (\bra{m} + \bra{n}) \Ket{\Psi}}\right| =\\
            = \frac{1}{2}\left|\sum_{k,\,l} \alpha_{kl} \bra{l} \left(\ketbra{m} + \ketbra{m}{n} + \ketbra{n}{m} + \ketbra{n}\right)\ket{k} \right| \\=
            \frac{1}{2}\left| \alpha_{mn} + \alpha_{nm} \right| > 0.
        \end{multline}
\end{proof}

Theorem~\ref{th1} guarantees that every element of the twin space forms a story with a suitably chosen measurement. This allows us to distinguish between two types of two-state vectors: those that form a story with \emph{every} measurement, and those that fail to form a story with \emph{at least one} measurement. The set of all two-state vectors $\Ket{\Psi} \in \mathcal{H} \otimes \mathcal{H}^\star$ that do not form a story with at least one measurement turns out to be a proper subset of the twin space $\mathcal{H} \otimes \mathcal{H}^\star$. This is ensured by the necessary condition of tracelessness, which we discuss below. As a consequence, two-state vectors that are not traceless must form a story with every measurement.

Taking a slightly different starting point, Reznik and Aharonov proposed to exclude traceless two-state vectors from the set of physical two-state vectors \cite{ra_1995}. However, such restriction seems unjustified for the case of ideal projective measurements, as from Theorem~\ref{th1} it follows that traceless vectors are physically achievable.

\begin{proposition}
    Every vector $\Ket{\Psi} \in \mathcal{H} \otimes \mathcal{H}^\star$ that does not form a story with some   measurement is \emph{traceless}:
    \begin{equation}
        \Tr{\Ket{\Psi}} = \sum_k \alpha_k \braket{\phi_k}{\psi_k} = 0.
    \end{equation}
\end{proposition}
\begin{proof}
    Suppose that the vector $\Ket{\Psi}$ does not form a story with a   measurement $\{\hat{P}_i\}$. For every $\hat{P}_i$, we have:
    \begin{equation}
        \Tr{\hat{P}_i \Ket{\Psi}} = \sum_k \alpha_k \bra{\phi_k} \hat{P}_i \ket{\psi_k} = 0.
    \end{equation}
    Hence, using the fact that the set $\{\hat{P}_i\}$ is a projective partition of unity and interchanging the order of summation, we obtain:
    \begin{multline}
        \Tr{\Ket{\Psi}} = \sum_k \alpha_k \braket{\phi_k}{\psi_k} = \sum_k \alpha_k \bra{\phi_k} \sum_i \hat{P}_i \ket{\psi_k} \\= 
        \sum_i \Tr{\hat{P}_i \Ket{\Psi}} = 0.
    \end{multline}
\end{proof}

The basic objective of the conventional two-state vector formalism is to compute the probabilities of distinct measurement outcomes. Therefore, in practical considerations, the measurement is assumed to be fixed. It is then possible to investigate the sets of states that either do or do not form a story with this measurement. In particular, one may ask whether these sets form linear subspaces within the twin space.

\begin{proposition}
    The subset $\mathcal{N}_{\{\hat{P}_i\}} \subset \mathcal{H} \otimes \mathcal{H}^\star$ of all states that do not form a story with a given measurement $\{\hat{P}_i\}$ is a linear subspace.
\end{proposition}
\begin{proof}
    Let us take $\Ket{\Psi_1}, \Ket{\Psi_2} \in \mathcal{N}_{\{\hat{P}_i\}}$. Then, for every index $i$ and $\alpha_1, \alpha_2 \in \mathbb{C}$, we have:
    \begin{multline}
        \Tr{\hat{P}_i (\alpha_1 \Ket{\Psi_1} + \alpha_2 \Ket{\Psi_2})} \\=
        \alpha_1 \Tr{\hat{P}_i \Ket{\Psi_1}} + \alpha_2 \Tr{\hat{P}_i \Ket{\Psi_2}} = 0.
    \end{multline}
    Thus, $\alpha_1 \Ket{\Psi_1} + \alpha_2 \Ket{\Psi_2} \in \mathcal{N}_{\{\hat{P}_i\}}$, so $\mathcal{N}_{\{\hat{P}_i\}}$ is a linear subspace.
\end{proof}

Since the subset $\mathcal{N}_{\{\hat{P}_i\}} \subset \mathcal{H} \otimes \mathcal{H}^\star$ is a linear subspace, it follows that the subset $\mathcal{H} \otimes \mathcal{H}^\star \setminus \mathcal{N}_{\{\hat{P}_i\}} \subset \mathcal{H} \otimes \mathcal{H}^\star$ of all vectors that form a story with the measurement $\{\hat{P}_i\}$ is not a linear subspace. Furthermore, it is meaningful to examine the dimensionality of the subspace $\mathcal{N}_{\{\hat{P}_i\}}$ in the case of finite-dimensional twin space. We show below that the dimension of $\mathcal{N}_{\{\hat{P}_i\}}$ is strictly smaller than the dimension of the whole twin space $\mathcal{H} \otimes \mathcal{H}^\star$ by exactly the number of elements in the set $\{\hat{P}_i\}$ (denoted by $\# \{\hat{P}_i\}$). Consequently, the subspace $\mathcal{N}_{\{\hat{P}_i\}}$ is of zero measure in the twin space.

\begin{proposition}
    The dimension of a subspace $\mathcal{N}_{\{\hat{P}_i\}} \subset \mathcal{H} \otimes \mathcal{H}^\star$ of all vectors that do not form a story with a given measurement $\{\hat{P}_i\}$ is $\dim \mathcal{N}_{\{\hat{P}_i\}} = (\dim \mathcal{H})^2 - \# \{\hat{P}_i\}$ for $\dim \mathcal{H} < \infty$.
\end{proposition}
\begin{proof}
    For the given measurement $\{\hat{P}_i\}$, let $\{\ket{\xi_j}\}$ denote an orthonormal basis of the space $\mathcal{H}$ introduced in Equation~\eqref{projection_decomposition}. Every vector $\Ket{\Psi} \in \mathcal{H} \otimes \mathcal{H}^\star$ can be decomposed in the product basis:
    \begin{equation}
        \Ket{\Psi} = \sum_{k,\,l} \alpha_{kl} \ket{\xi_k} \otimes \bra{\xi_l}.
    \end{equation}
    Suppose that a vector $\Ket{\Psi}$ does not form a story with the   measurement $\{\hat{P}_i\}$, which implies that for all indices $i$, we have:
    \begin{equation}
        \Tr{\hat{P}_i \Ket{\Psi}} = \sum_{k,\,l} \alpha_{kl} \bra{\xi_l} \left(\sum_{j \,\in\, J_i} \ketbra{\xi_j} \right) \ket{\xi_k} = \sum_{j \,\in\, J_i} \alpha_{jj} = 0.
    \end{equation}
    The above condition reduces the dimensionality of the subspace $\mathcal{N}_{\{\hat{P}_i\}}$ by exactly one. Since the dimension of the twin space $\mathcal{H} \otimes \mathcal{H}^\star$ is $(\dim \mathcal{H})^2$, we obtain $\dim \mathcal{N}_{\{\hat{P}_i\}} = (\dim \mathcal{H})^2 - \# \{\hat{P}_i\}$.
\end{proof}

To summarize, we have shown that every element of the twin space forms a story with an appropriately chosen measurement. A two-state vector may not form a story with some measurement only if it is traceless. The two-state vectors that are not traceless form a story with every measurement. When the  measurement is fixed, the twin space splits into two disjoint sets. The set comprising two-state vectors that form a story with the given measurement is not a linear subspace. However, the remaining set of two-state vectors that do not form a story with this measurement is a linear subspace of zero measure in the whole twin space.

\section{Distinguishability of two-state vectors}

Let us consider two ensembles of physical systems, each described by a different quantum state (either pure or mixed). By performing a measurement on either of these ensembles, one can obtain a probability distribution of distinct measurement outcomes. In the standard quantum theory, for any arbitrary pair of such ensembles, there exists a measurement that yields different probability distribution for each ensemble. This leads to the conclusion that all pairs of conventional quantum states are physically distinguishable (provided they are not the same). 

The situation for two-state vectors is different. There are pairs of two-state vectors that cannot be distinguished in the manner described above. This can be easily shown using the time symmetry inherent in the two-state vector formalism. Specifically, under the action of time reversal, two-state vectors transform as follows:
\begin{equation}
    \Ket{\Psi} = \sum_k \alpha_k \ket{\psi_k} \otimes \bra{\phi_k} \;\;\;\longmapsto\;\;\; \Ket{\Psi'} = \sum_k \bar{\alpha}_k \ket{\phi_k} \otimes \bra{\psi_k}.
\end{equation}
For every story $(\Ket{\Psi},\{\hat{P}_i\})$, the probabilities of distinct measurement outcomes are the same as in the story $(\Ket{\Psi'},\{\hat{P}_i\})$. Thus, the vectors $\Ket{\Psi}$ and $\Ket{\Psi'}$ are physically indistinguishable, although in general, they correspond to different elements of the twin space.

There are more non-trivial examples of indistinguishability. As an analogue to conventional mixed states, one can consider statistical ensembles of states described by two-state vectors. Two possible definitions of such ensembles have been discussed --- either using the weighted average of the probabilities given by Equation~\eqref{ABL} or introducing relative frequencies of different pre- and post-selection attempts \cite{s_2014}. We do not seek to rule out any of these definitions. Indeed, all the considerations that follow apply equally to both approaches.

It turns out that some two-state vectors cannot be distinguished from statistical mixtures of different two-state vectors. We provide an example of such a vector below.

\begin{example}\label{example_1}
    Let $\{\ket{0}, \ket{1}\}$ denote an orthonormal basis of the space $\mathcal{H}$. Consider the two-state vector:
    \begin{equation}
        \Ket{\Psi} = \ket{0} \otimes \bra{1} \in \mathcal{H} \otimes \mathcal{H}^\star.
    \end{equation}
    In the following, we show that there exists a statistical mixture of two-state vectors that replicates the probability distribution of distinct measurement outcomes in every story $(\Ket{\Psi}, \{\hat{P}_i\})$. Notice that, for every measurement $\{\hat{P}_1,\hat{P}_2\}$, the probabilities of distinct measurement outcomes in a story $(\Ket{\Psi},\{\hat{P}_1,\hat{P}_2\})$ are equal, since we have:
    \begin{multline}
        \Tr{\hat{P}_1 \Ket{\Psi}} = \bra{1} \hat{P}_1 \ket{0} = \bra{1} (\mathds{1} - \hat{P}_2) \ket{0} = - \bra{1} \hat{P}_2 \ket{0} \\= - \Tr{\hat{P}_2 \Ket{\Psi}}.
    \end{multline}
    Let us define a statistical mixture consisting of vectors $\Ket{\Psi_1} = \ket{0} \otimes \bra{0}$ and $\Ket{\Psi_2} = \ket{1} \otimes \bra{1}$ that occur with equal probabilities. Then, we get:
    \begin{multline}
        \Tr{\hat{P}_2 \Ket{\Psi_1}} = \Tr{(\mathds{1} - \hat{P}_1) \Ket{\Psi_1}} = 1 - \Tr{\hat{P}_1 \Ket{\Psi_1}} \\= \Tr{\hat{P}_1 \Ket{\Psi_2}},
    \end{multline}
    \begin{multline}
        \Tr{\hat{P}_2 \Ket{\Psi_2}} = \Tr{(\mathds{1} - \hat{P}_1) \Ket{\Psi_2}} = 1 - \Tr{\hat{P}_1 \Ket{\Psi_2}} \\= \Tr{\hat{P}_1 \Ket{\Psi_1}}.
    \end{multline}
    It follows that the probabilities of distinct measurement outcomes for this classical mixture are also equal. \qed
\end{example}

In conclusion, the general formula for calculating probabilities of measurement outcomes in the two-state vector formalism, given by Equation~\eqref{ABL}, leads to an unexpected result. Contrary to the case of conventional quantum state vectors, some two-state vectors are physically indistinguishable from other two-state vectors or their statistical mixtures. In the next section, we will further explore the implications of this fact.

\section{Separable and non-separable two-state vectors}

The tensor product structure of the twin space makes it possible to introduce the notion of a separable two-state vector. The definition of separability can be formulated analogously to the case of pure quantum states.
\begin{definition}
    A (pure) two-state vector $\Ket{\Phi} \in \mathcal{H} \otimes \mathcal{H}^\star$ is separable if and only if it is a tensor product of state vectors in each space $\mathcal{H}$ and $\mathcal{H}^\star$, and thus has the form:
    \begin{equation}
        \Ket{\Phi} = \ket{\psi} \otimes \bra{\phi} \in \mathcal{H} \otimes \mathcal{H}^\star.
    \end{equation}
\end{definition}
Separable two-state vectors correspond to pairs of evolved pre- and post-selected system states. This allows for a description of the system using concepts derived from the framework of standard quantum mechanics. 

Things get more interesting when one considers the possibility of obtaining non-separable two-state vectors, which have no straightforward analogue in conventional theory. The existence of non-separable two-state vectors holds a somewhat similar status (though not equivalent) to that of pure bipartite entangled states in quantum mechanics. A pure quantum state is entangled if and only if it is non-separable. Adopting a similar definition for two-state vectors implies the existence of a new kind of entanglement between the past and the future.

It would be problematic if there were no means of experimentally distinguishing non-separable two-state vectors from classical mixtures of separable two-state vectors. The following example, which has already been studied using a slightly different approach \cite{s_2014}, shows that at least some non-separable two-state vectors are, in fact, indistinguishable from classical mixtures of separable two-state vectors.

\begin{example}
    Let $\{\ket{0}, \ket{1}\}$ denote an orthonormal basis of the space $\mathcal{H}$. Consider the non-separable two-state vector:
    \begin{equation}
        \Ket{\Psi} = \frac{1}{\sqrt{2}} \left(\ket{0} \otimes \bra{0} + \ket{1} \otimes \bra{1}\right).
    \end{equation}
    In the following, we show that there exists a statistical mixture of separable two-state vectors that replicates the probability distribution of distinct measurement outcomes in every story $(\Ket{\Psi}, \{\hat{P}_i\})$. Notice that, for every measurement $\{\hat{P}_1,\hat{P}_2\}$, the probabilities of distinct measurement outcomes in a story $(\Ket{\Psi},\{\hat{P}_1,\hat{P}_2\})$ are equal. Let us define a statistical mixture consisting of separable vectors $\Ket{\Phi_1} = \ket{0} \otimes \bra{0}$ and $\Ket{\Phi_2} = \ket{1} \otimes \bra{1}$ that occur with equal probabilities. Analogically as in Example~\ref{example_1}, the probabilities of distinct measurement outcomes for this classical mixture are also equal. \qed
\end{example}

If, for every non-separable vector $\Ket{\Psi} \in \mathcal{H} \otimes \mathcal{H}^\star$, there existed a statistical mixture of separable two-state vectors that could replicate the probability distribution of distinct measurement outcomes in every story $(\Ket{\Psi}, \{\hat{P}_i\})$, the notion of entanglement in relation to two-state vectors would be rendered meaningless. Fortunately, not every non-separable two-state vector is indistinguishable from a classical mixture of separable two-state vectors. We show this by providing a counterexample.

\begin{example}
    Let $\{\ket{0}, \ket{1}, \ket{2}\}$ denote an orthonormal basis of the space $\mathcal{H}$. Consider the non-separable two-state vector:
    \begin{equation}
        \Ket{\Psi} = \frac{1}{\sqrt{3}} \left(\ket{0} \otimes \bra{0} + \ket{1} \otimes \bra{1} - \ket{2} \otimes \bra{2} \right).
    \end{equation}
    Every separable vector $\Ket{\Phi} \in \mathcal{H} \otimes \mathcal{H}^\star$ can be represented as:
    \begin{multline}
        \Ket{\Phi} = \left( \sum_k \alpha_k \ket{k}\right) \otimes \left(\sum_l \beta_l \bra{l}\right), \\ \text{where }\quad \sum_k \left|\alpha_k\right|^2 = \sum_l \left|\beta_l\right|^2 = 1.
    \end{multline}
    In the following, we prove by contradiction that no statistical mixture of such separable vectors can replicate the probability distribution of distinct measurement outcomes in every story $(\Ket{\Psi}, \{\hat{P}_i\})$. Let us define four sets of projective operators:
    \begin{equation}
        \left\{\begin{aligned}
            &\hat{P}_1^{(1)} = \ketbra{0},& &\hat{P}_2^{(1)} = \ketbra{1} + \ketbra{2}, \\
            &\hat{P}_1^{(2)} = \ketbra{1},& &\hat{P}_2^{(2)} = \ketbra{0} + \ketbra{2}, \\
            &\hat{P}_1^{(3)} = \ketbra{+},& &\hat{P}_2^{(3)} = \ketbra{-} + \ketbra{2}, \\
            &\hat{P}_1^{(4)} = \ketbra{+i},& &\hat{P}_2^{(4)} = \ketbra{-i} + \ketbra{2}.
        \end{aligned}\right.
    \end{equation}
    For all $i \in \{1,2,3,4\}$, the probability of a measurement outcome which corresponds to a projective operator $\hat{P}_2^{(i)}$ in a story $(\Ket{\Psi},\{\hat{P}_1^{(i)},\hat{P}_2^{(i)}\})$ is zero. Thus, every separable vector $\Ket{\Phi}$ belonging to the considered statistical mixture has to fulfill the following conditions:
    \begin{widetext}
    \begin{equation}\label{conditions}
        \left\{\begin{aligned}
            &\Tr{\hat{P}_2^{(1)} \Ket{\Phi}} = \alpha_1 \beta_1 + \alpha_2 \beta_2 = 0, \\
            &\Tr{\hat{P}_2^{(2)} \Ket{\Phi}} = \alpha_0 \beta_0 + \alpha_2 \beta_2 = 0, \\
            & \Tr{\hat{P}_2^{(3)} \Ket{\Phi}} = \frac{1}{2} \alpha_0 \beta_0 + \frac{1}{2} \alpha_1 \beta_1 - \frac{1}{2} \alpha_0 \beta_1 - \frac{1}{2} \alpha_1 \beta_0 + \alpha_2 \beta_2 = 0,\\
            & \Tr{\hat{P}_2^{(4)} \Ket{\Phi}} = \frac{1}{2} \alpha_0 \beta_0 + \frac{1}{2} \alpha_1 \beta_1 - \frac{i}{2} \alpha_0 \beta_1 + \frac{i}{2} \alpha_1 \beta_0 + \alpha_2 \beta_2 = 0.
        \end{aligned}\right.
    \end{equation}
    \end{widetext}
    Additionally, since the pair $(\Ket{\Phi},\{\hat{P}_1^{(1)},\hat{P}_2^{(1)}\})$ forms a story, we get:
    \begin{equation}
        \Tr{\hat{P}_1^{(1)} \Ket{\Phi}} = \alpha_0 \beta_0 \ne 0.
    \end{equation}
    Hence, the system of equations~\eqref{conditions} reduces to:
    \begin{equation}
        \left\{\begin{aligned}
            &\alpha_0 \beta_0 = \alpha_1 \beta_1 = - \alpha_2 \beta_2 \ne 0, \\
            & \alpha_0 \beta_1 = \alpha_1 \beta_0 = 0.
        \end{aligned}\right.
    \end{equation}
    From the first condition, we deduce that the coefficients $\alpha_0,$ $\beta_0,$ $\alpha_1,$ $\beta_1$ are nonzero. However, in this case, the second equation cannot be satisfied. Thus, the obtained system is inconsistent. \qed
\end{example}

We have demonstrated that for at least one non-separable two-state vector, it is impossible to find a classical mixture of separable two-state vectors that would yield the same probability distribution of distinct outcomes for every  measurement.

We conclude that the set of non-separable two-state vectors consists of two disjoint subsets, which differ in the distinguishability of their elements from classical mixtures of separable two-state vectors. This division is in stark contrast with the case of standard bipartite non-separable states, which are always distinguishable from classical mixtures of separable states. For the future investigation, we will call the distinguishable two-state vectors as \textit{strictly non-separable}. 

\begin{definition}
A two-state vector $\Ket{\Psi} \in \mathcal{H} \otimes \mathcal{H}^\star$ is \textit{strictly non-separable} if it can be experimentally distinguished from any (classical mixture of) separable two-state vectors.
\end{definition}

Strictly non-separable two-state vectors can be seen as the true counterparts of pure entangled states. One might consider them as the genuine manifestation of entanglement between the past and the future. An important question that arises is whether there exists a method for detecting such entanglement, analogous to conventional Bell inequalities or entanglement witnesses. This presents a compelling direction for future research.

\section{Discussion and conclusions}

We have introduced the concept of a story, which holds a fundamental physical significance within the framework of the two-state vector formalism. Using this concept, we thoroughly examined the structure of the twin space in relation to the possibility of forming stories with two-state vectors. We have demonstrated that every element of the twin space forms a story with an appropriately chosen measurement. Additionally, we formulated a sufficient condition for a two-state vector to form a story with every measurement. For a fixed measurement, we have shown that the set of two-state vectors that do not form a story with this measurement is a linear subspace of zero measure in the whole twin space, whereas the set comprising two-state vectors that form a story with this measurement is not a linear subspace.

In subsequent sections, we have examined the issue of the distinguishability of two-state vectors, concluding that some two-state vectors are indistinguishable from others or their statistical mixtures. In particular, we discussed an example of a non-separable two-state vector that is indistinguishable from a statistical mixture of separable two-state vectors. Fortunately, we have found that not every non-separable two-state vector is indistinguishable from such mixture, as we have demonstrated by providing a counterexample. This led us to formulate the definition of strictly non-separable vectors as the true counterparts of pure bipartite entangled states in conventional quantum theory.

Many other topics remain to be explored. Our conclusions serve merely as a starting point for a comprehensive investigation into the entanglement between the past and the future, adopting an approach different from those previously presented \cite{b_2004,rcgwf_2018,n_2017}. The non-separable two-state vectors that emerge within the structure of the twin space bear a resemblance to indefinite causal structures \cite{capv_2013} studied in the context of non-classical gravity \cite{zcpb_2019} and relativistic motion \cite{dmgmb_2020,dzcd_2023,nch_2018}. The correspondence between both descriptions has been addressed to some extent \cite{s_2017,ljqld_2024}, but it requires further attention. Finally, it would be beneficial to generalize our considerations to the relativistic framework, proposed either as a method of constructing a covariant quantum field theory of tachyons \cite{d_2023} or a theory of quantum gravity \cite{o_2015,r_2004}.

\begin{acknowledgments}
We would like to thank Rafał Demkowicz-Dobrzański for useful comments.
\end{acknowledgments}

\bibliography{stories}

\end{document}